\documentclass[11pt]{article}
\usepackage{amsmath,amssymb,amsthm}

\title{Aharonov--Bohm Type Arbitrage and Homological Obstructions in Financial Markets}

\author{
Takanori Adachi$^{1,2}$\thanks{This work was supported by JSPS KAKENHI Grant Number 24K04941.\\
Email: \texttt{taka.adachi@tmu.ac.jp}
}
\and
Keisuke Hara$^{2}$
}

\date{}

\newtheorem{definition}{Definition}
\newtheorem{remark}{Remark}
\newtheorem{proposition}{Proposition}

\newcommand{\Prob}{\mathbf{Prob}}

\newcommand{\Hol}{\mathrm{Hol}}
\newcommand{\Exp}{\mathcal{E}}

\begin{document}
\maketitle

\vspace{-1.5em}

\begin{center}
{\small
$^{1}$Graduate School of Management, Tokyo Metropolitan University\\
$^{2}$Association of Mathematical Finance Laboratory
}
\end{center}

\vspace{1.0em}

\begin{abstract}
\noindent
We introduce a simplicial and categorical formulation of
Aharonov--Bohm (AB) type arbitrage in filtered market systems.
Given a filtration modeled as a contravariant functor
$F : \mathcal T^{op} \to \mathbf{Prob},$
we consider the associated conditional expectation transport functor
$\mathcal E \circ F : \mathcal T^{op} \to \mathbf{Ban},$
and the canonical distortion
$dF(i) := (\mathcal E \circ F)(i)(1),$
which measures the failure of constant functions to be preserved under
non-measure-preserving transitions.

Motivated by the multiplicative transport structure of $dF$, 
we introduce a simplicial distortion operator defined recursively on the nerve 
$N_\bullet(\mathcal T)$ of the time category.
This construction describes recursively accumulated transported
distortions along composable chains of morphisms and leads naturally to
a notion of holonomy along loops.

We interpret non-trivial holonomy as a global observable invisible at
the level of individual transitions, 
analogous to the Aharonov--Bohm effect in physics.
This yields a notion of AB arbitrage, in which arbitrage opportunities
arise from global loop effects rather than local price discrepancies.

We further introduce simplicial admissibility conditions ensuring that
recursively accumulated distortions remain integrable, and show how
non-trivial holonomy can be translated into base-time observable self-financing
trading strategies through executable loop dynamics.
This establishes a connection between categorical holonomy structures
and economically realizable arbitrage.

The framework developed here suggests a global and homological
perspective on arbitrage theory, in which market inconsistencies are
encoded by recursively accumulated simplicial distortions and their
holonomy along loops in the underlying time category.
\end{abstract}

\section{Introduction}
\label{sec:introduction}

Arbitrage theory has traditionally been formulated in terms of local
price consistency and martingale structures on filtered probability
spaces. In classical financial mathematics, the absence of
arbitrage is closely related to the existence of equivalent martingale
measures, and market consistency is described through local stochastic
dynamics.

On the other hand, several geometric approaches to financial markets
have suggested that arbitrage phenomena may also possess global and
topological aspects. In particular, gauge-theoretic formulations of
finance have emphasized analogies between financial transport structures
and geometric holonomy phenomena. Motivated by the Aharonov--Bohm effect
in quantum physics, where globally nontrivial phase effects arise even
under locally flat gauge fields, we investigate a categorical framework
in which global arbitrage emerges from holonomy-like structures on time
categories.

The present paper develops a simplicial formulation of
Aharonov--Bohm arbitrage. The starting point is a filtration
\[
F: \mathcal{T}^{op} \to \Prob
\]
viewed as a contravariant functor from a time category
$\mathcal T$ to the category $\Prob$ of probability spaces and
null-preserving morphisms. Composing $F$ with conditional expectation
yields a transport functor
\[
\Exp\circ F,
\]
which transports integrable random variables backward along morphisms of
$\mathcal T$.

A central object of the paper is the simplicial distortion operator
\[
dF:N_\bullet(\mathcal T)\to \bigsqcup_{t\in\mathcal T}L^1(F(t)),
\]
defined recursively on simplices of the nerve
$N_\bullet(\mathcal T)$.
Unlike ordinary multiplicative cocycles, the values of $dF$ live in
different $L^1$-spaces depending on the base object of the simplex.
The recursive simplicial construction describes 
how local distortions are transported and recursively accumulated along composable chains of morphisms.
By transporting the accumulated distortion to a common probability space at each stage,
the construction yields a well-defined global observable associated with the chain.

The resulting simplicial distortion measures recursively accumulated
transported distortions along composable chains of morphisms.
For loops in the time category, this accumulation produces a holonomy
object
\[
\Hol(\gamma):=dF(\gamma),
\]
which plays the role of a global market observable.
Nontrivial holonomy is interpreted as a global observable
not represented by any single local transport,
but arising from the recursive accumulation of transported distortions along a loop.

The simplicial viewpoint introduced in this paper clarifies the global
nature of arbitrage phenomena. While the transport functor
$\Exp\circ F$ itself is functorial, recursively accumulated
distortions along loops may nevertheless generate nontrivial holonomy.
Thus AB arbitrage appears not as a local failure of transport, but as a
global effect produced by transported multiplicative accumulation.

Several geometric approaches to financial markets have suggested that
arbitrage phenomena may possess global and topological aspects.
Early gauge-theoretic viewpoints appeared in the work of Young
\cite{young_1999},
where foreign exchange markets were interpreted through lattice gauge
structures and holonomy-like effects.
More recently, geometric and spectral approaches to arbitrage developed
by Farinelli and Takada have further emphasized global geometric
structures underlying market dynamics

\noindent
\cite{FT_2021}, \cite{FT_2022}.

Motivated by these developments, as well as by the Aharonov--Bohm effect
in quantum physics, where globally nontrivial phase effects arise even
under locally consistent gauge fields, we investigate a categorical
framework in which global arbitrage emerges from recursively accumulated
holonomy-like structures on abstract time categories.

The paper is organized as follows.
Section \ref{sec:categoricalSetup}
reviews the categorical formulation of filtrations and conditional expectation transport.
Section \ref{sec:distortion}
introduces distortion operators on morphisms and discusses
their transported multiplicative structure.
Section \ref{sec:simpAdmi}
develops the simplicial distortion operator on the nerve
$N_\bullet(\mathcal T)$
and introduces simplicial admissibility.
Section \ref{sec:ABarbHolonomy}
defines holonomy and formulates Aharonov--Bohm arbitrage as a
global simplicial distortion phenomenon along loops in the time category.
Section \ref{sec:predTrading} 
shows how non-trivial holonomy can be translated into
base-time observable self-financing trading strategies through executable loop dynamics.
Section \ref{sec:example}
presents finite examples illustrating the emergence of
non-trivial holonomy and realizable AB arbitrage.
Finally, 
Section \ref{sec:admissibility}
introduces admissibility conditions ensuring that
loop-based holonomy structures correspond to economically realizable trading strategies.

\section{Categorical Setup}
\label{sec:categoricalSetup}

Let $\mathcal{T}$ be a small category representing time or information structure.

\begin{definition}[\cite{AR_2019}]
Let $\mathrm{Prob}$ denote the category whose objects are probability spaces
$(X,\mathcal{F},\mu)$ and whose morphisms
\[
\varphi : (X,\mathcal{F},\mu)\to (Y,\mathcal{G},\nu)
\]
are measurable maps such that $\varphi^{-1}$ preserves null sets,
i.e., $\nu(B)=0 \Rightarrow \mu(\varphi^{-1}(B))=0$.
This condition ensures that the associated conditional expectation operator
is well-defined via the Radon--Nikodym theorem.
\end{definition}

We model a filtration as a contravariant functor
\[
F : \mathcal{T}^{op} \to \mathrm{Prob}
\]
\cite{ANR_2020b}, \cite{Adachi_2025a}.
This formulation generalizes classical filtrations and stochastic processes,
in which time is linearly ordered, to more general categorical structures.


\medskip
\noindent
Next, define a functor
\[
\mathcal{E} : \mathrm{Prob} \to \mathrm{Ban},
\]
which assigns to each probability space its $L^1$ space, and to each morphism
$\varphi$ the conditional expectation operator
\[
\mathcal{E}(\varphi) : L^1(X) \to L^1(Y),
\]
characterized by the identity
\[
\int_B \mathcal{E}(\varphi)(f)\, d\nu
=
\int_{\varphi^{-1}(B)} f\, d\mu,
\qquad B\in\mathcal G,
\]
for $f \in L^1(X)$
\cite{AR_2019}.

\begin{remark}
A concrete representation of the conditional expectation
$
\mathcal{E}(\varphi)(f)
$
is a Radon-Nikodym derivative
\[
\mathcal{E}(\varphi)(f) 
		=
\frac{
	d (
		f \mu \circ \varphi^{-1}
	)
}{
	d \nu
} ,
\]
where for every 
$A \in \mathcal{F}$,
$
(f \mu)(A)
	:=
\int_A f d \mu.
$ 
\end{remark}

Thus, we obtain a composed functor
\[
\mathcal{E} \circ F : \mathcal{T}^{op} \to \mathrm{Ban}.
\]

\medskip
\noindent
Throughout the paper,
morphisms appearing in conditional expectation transport 
are regarded as arrows of the opposite category
$\mathcal T^{op}$.

\begin{remark}
A key feature of this setup is that, in general,
for an arrow $i$ in  $\mathcal{T}^{op}$, the operator $(\mathcal{E}\circ F)(i)$ 
does not preserve constant functions.
\end{remark}

\section{Distortion and $F$-Martingales}
\label{sec:distortion}

\begin{definition}[Distortion]
\label{defn:simpleDistortion}
For an arrow 
$i:t\to s$
in $\mathcal{T}^{op}$, define
\begin{equation}
\label{eq:defDf}
dF(i):=(\mathcal{E}\circ F)(i)(1_{F(t)}) \in L^1(F(s)).
\end{equation}
We call $dF(i)$ the \emph{distortion} associated with $i$.
\end{definition}

\begin{remark}
The distortion $dF(i)$ can be written as
\[
dF(i)
		=
\frac{
	d \big(
		\mu_t \circ F(i)^{-1}
	\big)
}{
	d \mu_s
} .
\]

\end{remark}

\begin{remark}
If the morphism $F(i)$ is measure-preserving, then
\[
(\mathcal{E}\circ F)(i)(1_{F(t)}) = 1_{F(s)},
\]
so $dF(i)$ is trivial. In general, $dF(i)$ captures the failure of constant
functions to be preserved under $(\mathcal{E}\circ F)(i)$.
\end{remark}

\begin{remark}
The present paper focuses on the distortion induced by the unit random
variable
\[
1_{F(t)}\in L^1(F(t)).
\]
Consequently, the distortion operator $dF$ captures the failure of
constant functions to be preserved under the transport functor
$
\Exp \circ F.
$
A more general operator-theoretic treatment involving arbitrary
integrable random variables
\[
X\in L^1(F(t))
\]
and their transported distortions will be investigated elsewhere.
\end{remark}

\medskip

\noindent
We now define a generalized martingale adapted to this distortion.

\begin{definition}[$F$-martingale]
\label{defn:Fmartingale}
A family $f=(f_t)$ with $f_t \in L^1(F(t))$ is called an \emph{$F$-martingale}
if for every arrow $i:t\to s$ in $\mathcal{T}^{op}$,
\[
(\mathcal{E}\circ F)(i)(f_t)=f_s \cdot dF(i).
\]
\end{definition}

\begin{remark}
In the classical case where all morphisms are measure-preserving, one has
$dF(i)=1$, and the above reduces to the usual martingale condition
\[
(\mathcal{E}\circ F)(i)(f_t)=f_s.
\]
Thus, $dF$ measures the deviation from classical martingale behavior.
\end{remark}

\medskip

\noindent
The distortion satisfies a multiplicative consistency relation.

\begin{proposition}[Multiplicativity of distortion]
\label{prop:multiplDistort}
For composable arrows 
$i:t \to s$
and
$j:u \to t$
in $\mathcal{T}^{op}$, we have
\begin{equation}
\label{eq:dFcomp}
dF(i \circ j)
=
(\mathcal{E}\circ F)(i)\bigl(dF(j)\bigr).
\end{equation}
\end{proposition}

\begin{proof}
By definition,
\[
dF(i \circ j)
=
(\mathcal{E}\circ F)(i \circ j)(1_{F(u)}).
\]
By functoriality of $\mathcal{E}\circ F$,
\[
(\mathcal{E}\circ F)(i \circ j)
=
(\mathcal{E}\circ F)(i)\circ (\mathcal{E}\circ F)(j),
\]
hence
\[
dF(i \circ j)
=
(\mathcal{E}\circ F)(i)\bigl((\mathcal{E}\circ F)(j)(1_{F(u)})\bigr)
=
(\mathcal{E}\circ F)(i)\bigl(dF(j)\bigr).
\]
\end{proof}

\begin{remark}
The identity (\ref{eq:dFcomp})
suggests that dF behaves analogously to a transported multiplicative cocycle.
\end{remark}

\begin{remark}[Price vector potential]
The distortion $dF$ may be viewed as a multiplicative analogue of a
vector potential. It assigns to each arrow a local multiplicative
distortion, and its composition rule reflects the transport of this
distortion along composable paths.
\end{remark}

\section{Simplicial Admissibility}
\label{sec:simpAdmi}

Recall that $[n]$ denotes the totally ordered category
\[
0 \longrightarrow 1 \longrightarrow \cdots \longrightarrow n.
\]
The \emph{nerve} of a category $\mathcal T$ is the simplicial set
\[
N_\bullet(\mathcal T),
\qquad
N_n(\mathcal T)=\operatorname{Fun}([n],\mathcal T),
\]
whose $n$-simplices are functors from $[n]$ to $\mathcal T$.

Equivalently, an $n$-simplex
\[
\sigma\in N_n(\mathcal T)
\]
may be identified with a composable chain of morphisms
\[
t_0
\xrightarrow{i_1}
t_1
\xrightarrow{i_2}
\cdots
\xrightarrow{i_n}
t_n
\]
in $\mathcal T$,
which is also considered as a chain
\[
t_0
\xleftarrow{i_1}
t_1
\xleftarrow{i_2}
\cdots
\xleftarrow{i_n}
t_n
\]
in $\mathcal T^{op}$,

In particular,
\[
N_0(\mathcal T)
\]
consists of the objects of $\mathcal T$,
\[
N_1(\mathcal T)
\]
consists of the morphisms of $\mathcal T$,
and
\[
N_2(\mathcal T)
\]
consists of composable pairs of morphisms
\[
t_0 \to t_1 \to t_2.
\]

More generally, higher simplices encode composable chains of morphisms of increasing length. 
In the present context, such simplices may be viewed as paths in the time category $\mathcal T$,
along which distortions are transported and accumulated.

\medskip
\noindent
Proposition~\ref{prop:multiplDistort} suggests that the distortion
operator \(dF\) should naturally be regarded as an object associated not
only with morphisms of the time category \(\mathcal T\), but also with
composable chains of morphisms, namely simplices of the nerve
\(N_\bullet(\mathcal T)\).


For $n\ge 1$, let
\[
\delta_0^n:[n-1]\to[n]
\]
denote the coface map omitting the vertex $0$.
If
\[
\sigma=
\left(
t_0 \xrightarrow{i_1} t_1 \xrightarrow{i_2}
\cdots \xrightarrow{i_n} t_n
\right),
\]
then
\[
\sigma\circ\delta_0^n
=
\left(
t_1 \xrightarrow{i_2} t_2 \xrightarrow{i_3}
\cdots \xrightarrow{i_n} t_n
\right)
\]
is the tail simplex obtained by removing the first vertex and the first arrow.

\begin{definition}
\label{def:simplicialDF}
Let
\[
F:\mathcal T^{op}\to\Prob
\]
be a filtration.
For a simplex
$
\sigma\in N_n(\mathcal T),
$
define the simplicial distortion operator
$dF$
by
\begin{equation}
\label{eq:simplicialDF}
dF(\sigma)
	:=
\begin{cases}
1_{F\big(\sigma(0)\big)}
		& \text{if }
n = 0,
		\\
(\Exp \circ F)\big( \sigma(0 \to 1) \big)
\Big(1_{F\big(\sigma(1)\big)}\Big)
		& \text{if }
n = 1,
		\\
(\Exp \circ F)\big( \sigma(0 \to 1) \big)
\Big(1_{F\big(\sigma(1)\big)}\Big)
	\, \cdot \,
(\Exp \circ F)\big( \sigma(0 \to 1) \big)
\big(
	dF(
		\sigma \circ \delta^n_0
	)
\big) 
		& \text{if }
n > 1.
\end{cases}
\end{equation}
\end{definition}

\begin{remark}
The above construction is the multiplicative analogue of parallel transport along a path,
in which distortions are propagated and accumulated along the path.
The recursive definition of the simplicial distortion operator
involves transported multiplicative accumulation of distortions along the path $\sigma$.
It is also a generalization of the operator $dF$ over single morphisms defined in 
Section \ref{sec:distortion}.
\end{remark}

\begin{remark}
Let
\[
\sigma=
\left(
t_0\xrightarrow{i_1}t_1\xrightarrow{i_2}
\cdots\xrightarrow{i_n}t_n
\right)
\in N_n(\mathcal T).
\]
The recursive definition of $dF(\sigma)$ should not be confused with
the naive product of independently transported local distortions. Indeed,
conditional expectation operators are linear but not multiplicative in
general.

For instance, if
\[
\sigma=
\left(
t_0\xrightarrow{i_1}t_1\xrightarrow{i_2}t_2
\xrightarrow{i_3}t_3
\right),
\]
then
\[
\begin{aligned}
dF(\sigma)
&=
dF(i_1)\cdot
(\Exp\circ F)(i_1)
\left(
dF(i_2)\cdot
(\Exp\circ F)(i_2)\bigl(dF(i_3)\bigr)
\right) \
\\&=
(\Exp\circ F)(i_1)(1_{F(t_1)})\cdot
(\Exp\circ F)(i_1)
\left(
(\Exp\circ F)(i_2)(1_{F(t_2)})
\cdot
(\Exp\circ F)(i_2)
\bigl((\Exp\circ F)(i_3)(1_{F(t_3)})\bigr)
\right).
\end{aligned}
\]
Thus the simplicial distortion is a nested transported multiplicative
accumulation,
not a simple product of separately transported factors.
\end{remark}

\begin{remark}
The simplicial distortion $dF(\sigma)$ depends on the simplex
$\sigma$, not merely on the underlying endpoints.  
In particular, it records how local distortions are sampled, transported, and accumulated
along the composable chain represented by $\sigma$.  
Thus the present notion of distortion should be understood as a simplicial accumulated distortion
rather than as a curvature form in the usual differential geometric sense.

This distinction is important.  
In ordinary differential geometry,
parallel transport depends on a path, but not on an auxiliary choice of intermediate subdivision of that path.  
By contrast, the present construction is sensitive to the simplicial decomposition itself.
 This reflects the fact that the theory is built from the categorical composition of market transitions
and the corresponding conditional expectation transports. 
It would be interesting to investigate whether such accumulated
distortions admit a more intrinsic geometric formulation.
\end{remark}

\medskip
\noindent
Since products of \(L^1\)-random variables are not generally
\(L^1\)-integrable, 
we impose the following admissibility condition to ensure that all simplicial distortions are well-defined in the corresponding \(L^1\)-spaces.

\begin{definition}
A filtration
$
F:\mathcal T^{op}\to\Prob
$
is called \emph{simplicially admissible} if
\[
dF(\sigma)
	\in
L^1\big(F(\sigma(0))\big)
\]
for every simplex
$
\sigma
\in N_\bullet(\mathcal T).
$
\end{definition}

\medskip
\noindent
In the rest of this paper, we assume the filtration $F$ we consider is simplicially admissible.

\section{AB Arbitrage and Holonomy}
\label{sec:ABarbHolonomy}

We now define the global effect associated with the distortion $dF$ along loops in $\mathcal{T}$.


\begin{definition}[Holonomy]
Let $\gamma$ be a composable sequence of arrows in $\mathcal{T}^{op}$:
\[
\gamma = (i_1, i_2, \dots, i_n),
\]
where
\[
i_k : t_{k} \to t_{k-1}, \qquad k=1,\dots,n,
\]
and $t_n = t_0$, so that $\gamma$ forms a loop.

The \emph{holonomy} of $\gamma$ is defined by
\begin{equation}
\label{eq:holonomy}
\mathrm{Hol}(\gamma)
	:=
dF(\gamma)
	\in
L^1(F(t_0)) .
\end{equation}
\end{definition}

\begin{remark}
The holonomy $\Hol(\gamma)$ is obtained by recursively transporting
and multiplicatively accumulating local distortions along the loop
$\gamma$. Although the transport functor $\Exp\circ F$ itself is
functorial, the cumulative effect of transported distortions along
loops may produce nontrivial global holonomy phenomena.
\end{remark}

\medskip
\noindent
We now define AB arbitrage.

\begin{definition}[Weak AB-arbitrage / non-trivial holonomy effect]
We say that the system exhibits weak AB-arbitrage (or a non-trivial holonomy effect) 
if there exists a loop $\gamma$ such that
\begin{equation}
\label{eq:weakABarb}
\mu_{t_0}\bigl(\Hol(\gamma)\neq 1\bigr) > 0,
\end{equation}
where
$\mu_t$
is the probability measure of
$F(t)$.
\end{definition}

\begin{remark}
Weak AB-arbitrage represents a global inconsistency in the system 
that is not detectable at the level of individual transitions.
It appears only as a global inconsistency detected through the holonomy along a loop.

It does not by itself imply the existence of a realizable arbitrage opportunity.
\end{remark}

\begin{definition}[AB-arbitrage]
We say that the system exhibits \emph{Aharonov--Bohm (AB) arbitrage}
if there exists a loop $\gamma$ such that
\begin{equation}
\label{eq:ABarb}
\mu_{t_0}\bigl(\Hol(\gamma) > 1\bigr) > 0.
\end{equation}
\end{definition}

\begin{remark}
The condition 
(\ref{eq:ABarb})
represents the existence of a potential arbitrage opportunity associated with the loop $\gamma$.
In the next section, we show that, under suitable admissibility conditions,
this can be realized as an actual arbitrage via a base-time observable self-financing trading strategy.
\end{remark}

\medskip
\noindent
The following proposition shows that trivial distortion yields no holonomy.

\begin{proposition}
If $dF(i)=1_{F(t_k)}$ for all arrows $i: t_{k+1}\to t_{k}$ in $\mathcal T^{op}$, then
\[
\mathrm{Hol}(\gamma)=1_{F(t_0)}
\]
for every loop $\gamma$.
\end{proposition}
\begin{proof}
We prove the statement by induction on the length of the loop $\gamma$.

If $\gamma$ consists of a single arrow, then the statement follows directly from the assumption.

Assume the statement holds for loops of length $n-1$.
Let
\[
\gamma=(i_1,\dots,i_n)
\]
be composable morphisms in $\mathcal T^{op}$.

By Definition~\ref{def:simplicialDF},
\[
\Hol(\gamma)
=
dF(\gamma)
=
dF(i_1)\cdot
(\Exp\circ F)(i_1)
\bigl(
dF(\gamma\circ\delta_0^n)
\bigr).
\]
Since
\[
dF(i_1)=1_{F(t_0)},
\]
and by the induction hypothesis
\[
dF(\gamma\circ\delta_0^n)=1_{F(t_1)},
\]
we obtain
\[
\Hol(\gamma)=1_{F(t_0)}.
\]
\end{proof}

\medskip

\noindent
Thus, non-trivial holonomy reflects a genuine global inconsistency in the system.

\medskip

\noindent
We now turn to the construction of admissible trading strategies that realize such
potential arbitrage opportunities.

\section{Base-Time Observable Trading Strategy and Execution of Loops}
\label{sec:predTrading}

We now show how a non-trivial holonomy can be translated
into a base-time observable trading strategy.

\medskip

\noindent
Let $\gamma$ be a loop based at $t_0$, and let
\[
\mathrm{Hol}(\gamma) \in L^1(F(t_0))
\]
be its holonomy as defined in 
Section \ref{sec:ABarbHolonomy}.

\medskip

\noindent
Since $\mathrm{Hol}(\gamma)$ is measurable with respect to $F(t_0)$,
the decision whether to execute the loop can be made at the base time.

\medskip

Assume that the system exhibits AB arbitrage,
that is,
there exists a loop $\gamma$ such that
\[
\mu_{t_0}(\Hol(\gamma)>1)>0,
\]
and define the position
\begin{equation}
\label{eq:ABpositionTheta}
\theta^{\mathrm{AB}}_{\gamma, t_0}
	:=
1_{\{\Hol(\gamma)>1\}}.
\end{equation}

\medskip

\noindent
The trading strategy is as follows:
\begin{itemize}
  \item if $\Hol(\gamma) > 1$, execute the loop $\gamma$,
  \item otherwise, do nothing.
\end{itemize}

\begin{remark}[Execution of a loop]
To execute a loop $\gamma$, the trader commits at time $t_0$ to a sequence
of trades corresponding to the arrows of $\gamma$. Each trade is valued at
time $t_0$ via the functor $\mathcal{E}\circ F$, so that the entire strategy is
evaluated without using future information.

The resulting terminal wealth, expressed in units at time $t_0$, is given
by the holonomy $\mathrm{Hol}(\gamma)$. Thus, the loop represents a
self-financing round-trip strategy whose gain is determined by the
global distortion along $\gamma$.
\end{remark}


\noindent
This decision is $F(t_0)$-measurable.

\medskip

\noindent
Assuming that the sequence of trades corresponding to $\gamma$
can be executed using the observed distortions $dF(i_k)$,
the resulting gain from one unit of initial capital is
\begin{equation}
\label{eq:ABgainOneUnit}
V_{t_0}^{\mathrm{AB}}
=
1
	+
\theta^{\mathrm{AB}}_{\gamma, t_0}
\bigl(\mathrm{Hol}(\gamma)-1\bigr).
\end{equation}

\medskip

\begin{proposition}
\label{prop:ABarb}
For the strategy
$
\theta^{\mathrm{AB}}_{\gamma, t_0}
$,
we have
\[
V_{t_0}^{\mathrm{AB}} \ge 1,
\]
with strict inequality on the set
\[
\{\Hol(\gamma) > 1\}.
\]
\end{proposition}
\begin{proof}
Immediate from the definition of 
$\theta^{\mathrm{AB}}_{\gamma, t_0}$.
\end{proof}

\noindent
Therefore, if
$
\mu_{t_0}\bigl(\mathrm{Hol}(\gamma) > 1\bigr)>0,
$
the position
$
\theta^{\mathrm{AB}}_{\gamma, t_0}
$
yields an executable arbitrage opportunity.

\medskip

\noindent
Next, assume that the system exhibits weak AB arbitrage.
Fix $\varepsilon \ge 0$, and define
\begin{equation}
\label{eq:weakABpositionTheta}
\theta^{\mathrm{wAB}}_{\gamma, t_0}
:=
\mathbf{1}_{\{\mathrm{Hol}(\gamma)>1+\varepsilon\}}
-
\mathbf{1}_{\{\mathrm{Hol}(\gamma)<1-\varepsilon\}}.
\end{equation}

\noindent
The trading strategy is as follows:
\begin{itemize}
\item if $\mathrm{Hol}(\gamma)>1+\varepsilon$, execute the loop $\gamma$,
\item if $\mathrm{Hol}(\gamma)<1-\varepsilon$, execute the reverse loop $\gamma^{-1}$,
\item otherwise, do nothing.
\end{itemize}

\medskip

\noindent
Under suitable assumptions on reverse execution, the resulting wealth satisfies
\[
V_{t_0}^{\mathrm{wAB}} \ge 1,
\]
with strict inequality on the set
\[
\{|\Hol(\gamma)-1|>\varepsilon\}.
\]

\begin{remark}[Reverse execution]
The reverse loop $\gamma^{-1}$ is obtained by reversing the order of trades
and inverting each transaction. This does not require the existence of
categorical inverses in $\mathcal T$, but rather corresponds to the
availability of opposite trades in the market.

Formally, executing $\gamma^{-1}$ amounts to applying the inverse sequence
of exchanges, 
so that its effect 
can be interpreted as approximately corresponding to the reciprocal
$
\mathrm{Hol}(\gamma)^{-1}.
$
\end{remark}

\begin{remark}
In both strategies,
the key feature is that the entire decision is made using information
available at time $t_0$. No future information is required.
Thus, the arbitrage arises from a global observable encoded in the
holonomy, rather than from anticipative behavior.
\end{remark}

\medskip

\noindent
This shows that non-trivial holonomy is not merely a structural or
cohomological artifact, but can be translated into an economically
meaningful trading strategy, thereby realizing arbitrage.

\section{Illustrative Examples}
\label{sec:example}

In this section,
we present two finite examples illustrating how non-trivial holonomy arises from non-measure-preserving transitions.

\medskip

\noindent
Let $\mathcal{T}$ be a category with objects $t_0,t_1,t_2$ and arrows
\[
t_0 \xrightarrow{i_1} t_1 \xrightarrow{i_2} t_2 \xrightarrow{i_3} t_0,
\]
forming a loop $\gamma=(i_1,i_2,i_3)$.

\subsection{A simple finite example}

Define a contravariant functor $F:\mathcal{T}^{\mathrm{op}}\to \mathrm{Prob}$ by
\[
F(t_0)=\bigl(\{0,1\}, 2^{\{0,1\}}, \mu_0\bigr), \quad
F(t_1)=\bigl(\{\ast\}, 2^{\{\ast\}}, \mu_1\bigr), \quad
F(t_2)=\bigl(\{0,1\}, 2^{\{0,1\}}, \mu_2\bigr),
\]
where
\[
\mu_0(\{0\})=\mu_0(\{1\})=\tfrac12, \qquad
\mu_1(\{\ast\})=1, \qquad
\mu_2(\{0\})=\tfrac14,\ \mu_2(\{1\})=\tfrac34.
\]

The corresponding measurable maps are defined by
\[
F(i_1):\{\ast\}\to\{0,1\},\ \ast\mapsto 1, \quad
F(i_2):\{0,1\}\to\{\ast\}, \quad
F(i_3):\{0,1\}\to\{0,1\},\ x\mapsto x.
\]

Then,
the associated distortions are given by
\[
dF(i_1)=2 \cdot 1_{\{1\}}, \quad
dF(i_2)=1, \quad
dF(i_3)=\frac{d\mu_0}{d\mu_2},
\]
so that
\[
dF(i_3)(0)=2, \qquad dF(i_3)(1)=\tfrac{2}{3}.
\]

A direct computation using the definition of holonomy yields
\[
\operatorname{Hol}(\gamma)=4 \cdot 1_{\{1\}}.
\]

In particular,
\[
\mu_0\bigl(\operatorname{Hol}(\gamma)>1\bigr)=\tfrac12>0,
\]
and the system exhibits AB arbitrage.
Moreover, using the strategy defined in 
Section \ref{sec:predTrading},
this yields a base-time observable arbitrage opportunity.

\subsection{A stronger finite example}

\noindent
In contrast to the previous example, 
this construction avoids loss of mass along the loop, ensuring that the resulting holonomy is pointwise nonnegative.
In other words,
we present a finite example in which the holonomy itself satisfies
\[
\mu_{t_0}\bigl(\Hol(\gamma)\ge 1\bigr)=1,
\qquad
\mu_{t_0}\bigl(\Hol(\gamma)>1\bigr)>0.
\]
Thus, in this case, executing the loop $\gamma$ directly yields an arbitrage opportunity.

\medskip

\noindent
We define a filtration $F:T^{\mathrm{op}}\to\Prob$ by
\[
F(t_0)=\bigl(\{0,1\},2^{\{0,1\}},\mu_0\bigr),\quad
F(t_1)=\bigl(\{a,b,c\},2^{\{a,b,c\}},\mu_1\bigr),\quad
F(t_2)=\bigl(\{u,v,w\},2^{\{u,v,w\}},\mu_2\bigr),
\]
with
\[
\mu_0(0)=\tfrac14,\ \mu_0(1)=\tfrac34,\quad
\mu_1(a)=\tfrac14,\ \mu_1(b)=\tfrac14,\ \mu_1(c)=\tfrac12,\quad
\mu_2(u)=\tfrac14,\ \mu_2(v)=\tfrac14,\ \mu_2(w)=\tfrac12.
\]

The measurable maps are given by
\[
F(i_1): a\mapsto 0,\ b\mapsto 1,\ c\mapsto 1,\quad
F(i_2): u\mapsto a,\ v\mapsto c,\ w\mapsto b,\quad
F(i_3): 0\mapsto u,\ 1\mapsto w.
\]

\medskip

We compute the associated distortions.

First, $(\mu_1\circ F(i_1)^{-1})=\mu_0$, hence
\[
dF(i_1)=1.
\]

Next,
\[
(\mu_2\circ F(i_2)^{-1})(a)=\tfrac14,\quad
(\mu_2\circ F(i_2)^{-1})(b)=\tfrac12,\quad
(\mu_2\circ F(i_2)^{-1})(c)=\tfrac14,
\]
so
\[
dF(i_2)(a)=1,\quad dF(i_2)(b)=2,\quad dF(i_2)(c)=\tfrac12.
\]

Finally,
\[
(\mu_0\circ F(i_3)^{-1})(u)=\tfrac14,\quad
(\mu_0\circ F(i_3)^{-1})(v)=0,\quad
(\mu_0\circ F(i_3)^{-1})(w)=\tfrac34,
\]
hence
\[
dF(i_3)(u)=1,\quad dF(i_3)(v)=0,\quad dF(i_3)(w)=\tfrac32.
\]

\medskip

We now compute the holonomy.

Since $h_3(\gamma)=dF(i_3)$, we obtain
\[
(\Exp \circ F)(i_2)(h_3)(a)=1,\quad
(\Exp \circ F)(i_2)(h_3)(b)=3,\quad
(\Exp \circ F)(i_2)(h_3)(c)=0,
\]
and therefore
\[
h_2(a)=1,\quad h_2(b)=6,\quad h_2(c)=0.
\]

Since $dF(i_1)=1$, we conclude
\[
\Hol(\gamma)(0)=1,\qquad
\Hol(\gamma)(1)=2,
\]
that is,
\[
\Hol(\gamma)=1_{\{0\}}+2 \cdot 1_{\{1\}}.
\]

\medskip

In particular,
\[
\mu_0\bigl(\Hol(\gamma)\ge 1\bigr)=1,
\qquad
\mu_0\bigl(\Hol(\gamma)>1\bigr)=\tfrac34>0.
\]

Thus, the loop $\gamma$ itself yields an arbitrage opportunity without requiring selective execution.

\section{Admissibility and AB Arbitrage}
\label{sec:admissibility}

In the previous sections, we introduced holonomy as a global invariant
associated with the distortion $dF$, and showed how it can be used to define
a base-time observable trading strategy. 
However, not every loop corresponds to a financially realizable trading strategy.

In this section, we formalize the notion of admissibility, which ensures that
a loop can be interpreted as an executable and economically meaningful
trading strategy.

\subsection{Admissible loops}

The construction in
Section \ref{sec:predTrading}
shows that holonomy is observable at the base time.
Whether this observable can be converted into an actual trading strategy
depends on additional market conditions,
which are formalized below as admissibility.

Let $\gamma=(i_1,\dots,i_n)$ be a loop based at $t_0$.

\begin{definition}[Admissible loop]
We say that a loop $\gamma$ is admissible 
if the following conditions hold:

\begin{enumerate}
\item \textbf{Observability:} The holonomy $\mathrm{Hol}(\gamma)$ belongs to
$L^1(F(t_0))$, and is measurable with respect to $F(t_0)$.

\item \textbf{Executability:} Each arrow $i_k : t_{k-1}\to t_k$ corresponds
to an available market transaction.

\item \textbf{Composability:} The sequence of transactions along $\gamma$
can be composed sequentially, so that the output of each step can be used
as the input of the next.

\item \textbf{Self-financing:} The loop can be executed without injecting or
withdrawing additional capital after time $t_0$.


\item \textbf{Reverse executability (optional):}
The reverse sequence of trades
$\gamma^{-1}$
is also executable 
when reverse execution is required.

\end{enumerate}
\end{definition}

\begin{remark}
The above conditions separate the geometric structure of the cocycle from
the financial structure of the market. In particular, admissibility depends
on market microstructure and is not determined purely by the category
$\mathcal{T}$.
\end{remark}

\subsection{Self-financing interpretation}

We give a more explicit description of the self-financing condition.

\medskip

\noindent
Let $V_0=1$ be the initial wealth at time $t_0$.
Define recursively
\[
V_k = V_{k-1} \cdot \widetilde d_k,
\qquad k=1,\dots,n,
\]
where $\widetilde d_k$ denotes the valuation at time $t_0$ of the $k$-th leg
of the loop.

Then, the terminal wealth obtained from one unit of initial capital is
\[
V_n = \mathrm{Hol}(\gamma).
\]

Thus, the holonomy represents the terminal wealth of a self-financing
round-trip strategy.

\subsection{Admissible AB arbitrage}

We can now refine the notion of AB arbitrage.

\begin{definition}[Admissible AB arbitrage]
We say that the system exhibits \emph{admissible AB arbitrage} if there exists
an admissible loop $\gamma$ based at $t_0$ such that
\[
\mu_{t_0}\bigl(\mathrm{Hol}(\gamma) > 1\bigr) > 0.
\]
\end{definition}
This condition ensures that the strategy
$
\theta^{\mathrm{AB}}_{\gamma,t_0}
$
constructed in 
Section \ref{sec:predTrading}
yields an arbitrage opportunity.

If both $\gamma$ and its reverse $\gamma^{-1}$ are admissible, we may equivalently
consider the sign-based strategy
$
\theta^{\mathrm{wAB}}_{\gamma,t_0}
$
defined in
(\ref{eq:weakABpositionTheta})
for $\varepsilon \ge 0$.

\begin{proposition}
Let $\gamma$ be an admissible loop based at $t_0$.
If
\[
\mu_{t_0}\bigl(|\mathrm{Hol}(\gamma)-1|>\varepsilon\bigr)>0,
\]
then the strategy defined by
$\theta^{\mathrm{wAB}}_{\gamma,t_0}$
yields a base-time observable
self-financing trading strategy with nonnegative payoff and positive
probability of strict gain.
\end{proposition}

\begin{proof}
The result follows from the construction in 
Section \ref{sec:predTrading},
together with admissibility, which guarantees that the strategy is executable and
self-financing.
\end{proof}

\subsection{Discussion}

The notion of admissibility highlights a key structural feature of AB
arbitrage:
holonomy captures a global inconsistency in the system,
while admissibility ensures that this inconsistency can be monetized.

Thus, AB arbitrage arises from the interaction between
the categorical structure of the filtration and the economic structure
of the market.

\section{Conclusion}
\label{sec:conclusion}

In this paper, we developed a simplicial formulation of
Aharonov--Bohm arbitrage based on categorical filtrations and recursive
distortion transport.
Viewing a filtration
\[
F:\mathcal T^{op}\to\Prob
\]
as a contravariant functor on a time category $\mathcal T$, we
constructed a simplicial distortion operator
\[
dF:N_\bullet(\mathcal T)\to
\bigsqcup_{t\in\mathcal T}L^1(F(t))
\]
defined recursively on the nerve of $\mathcal T$.

The simplicial viewpoint clarifies that arbitrage phenomena may arise as
global effects produced by recursively accumulated transported
distortions along loops.
The resulting holonomy
\[
\Hol(\gamma):=dF(\gamma)
\]
provides a categorical analogue of geometric holonomy and serves as an
observable of global market inconsistency.
In this sense, AB arbitrage is interpreted not as a local failure of
transport, but as a global phenomenon emerging from multiplicative
accumulation along loops in the time category.

The recursive simplicial formulation provides a mathematically consistent framework
for transporting and recursively accumulating local distortions along composable chains.
In particular,
transporting accumulated distortions to a common probability space at each stage
naturally resolves the typing inconsistency of naive multiplicative constructions.
This leads naturally to the notion of simplicial admissibility,
guaranteeing integrability of recursively accumulated distortions.

The simplicial structure further suggests deeper relations between
global arbitrage phenomena and homotopy-theoretic structures associated
with the time category.
Indeed, the nerve $N_\bullet(\mathcal T)$ admits a geometric
realization
\[
B\mathcal T := |N_\bullet(\mathcal T)|,
\]
the classifying space of $\mathcal T$.
From this viewpoint, loops in the time category may be interpreted as
categorical analogues of topological loops, suggesting possible future
connections between AB arbitrage, holonomy representations, and the
topology of $B\mathcal T$.

We hope that the simplicial and categorical viewpoint developed here
will contribute to a broader geometric understanding of arbitrage and
global market structures.


\end{document}